\documentclass[11pt,letterpaper]{article}
\usepackage[letterpaper,margin=1in]{geometry}
\usepackage[utf8]{inputenc}

\usepackage{amsmath}
\usepackage{amsfonts}
\usepackage{amssymb}
\usepackage{amsthm}

\usepackage{thmtools}

\declaretheorem{theorem}
\declaretheorem[sibling=theorem]{lemma}

\declaretheorem[sibling=theorem,style=definition]{definition}

\usepackage{url,hyperref}
\usepackage{tikz}
\usetikzlibrary{fit,positioning}
\usepackage{xspace}
\usepackage{stmaryrd}

\usepackage[all]{nowidow}

\newcommand{\Oh}[1]{\mathcal{O}{\left(#1\right)}}
\newcommand{\Ot}[1]{\widetilde{\mathcal{O}}{\left(#1\right)}}
\newcommand{\Th}[1]{\Theta{\left(#1\right)}}

\let\eps\varepsilon
\let\leq\leqslant
\let\geq\geqslant

\newcommand{\mono}{\ensuremath{\mathsf{AE\text{-}Mono}\Delta}\xspace}
\newcommand{\sparse}{\ensuremath{\mathsf{AE\text{-}Sparse}\Delta}\xspace}
\newcommand{\minw}{\ensuremath{\mathsf{MinWitness}}\xspace}
\newcommand{\minmax}{\ensuremath{\mathsf{Min\text{-}Max}}\xspace}
\newcommand{\uapsp}{\ensuremath{\mathsf{UnweightedAPSP}}\xspace}
\newcommand{\coin}{\ensuremath{\mathsf{CoinChange}}\xspace}
\newcommand{\monoconv}{\ensuremath{\mathsf{MonoConvolution}}\xspace}

\title{Monochromatic Triangles, Intermediate Matrix Products,\\and Convolutions\thanks{We would like to thank Amir Abboud for fruitful discussions at an early stage of our research. Part of the research was done when the second author was visiting MIT. A preliminary version of this paper was presented at ITCS 2020.}}

\author{Andrea Lincoln\thanks{Partially supported by NSF Grant CCF-1909429.}\\MIT\\\texttt{andreali@mit.edu} \and Adam Polak\thanks{Partially supported by the National Science Center, Poland under grants 2017/27/N/ST6/01334 and 2018/28/T/ST6/00305.}\\Jagiellonian Univeristy\\\texttt{polak@tcs.uj.edu.pl} \and Virginia Vassilevska Williams\thanks{Supported by an NSF CAREER Award, NSF Grants CCF-1528078, CCF-1514339 and CCF-1909429, a BSF Grant BSF:2012338, a Google Research Fellowship and a Sloan Research Fellowship.}\\MIT\\\texttt{virgi@mit.edu}}

\date{}

\begin{document}

\maketitle

\begin{abstract}
The most studied linear algebraic operation, matrix multiplication, has surprisingly fast $O(n^\omega)$ time algorithms for $\omega<2.373$. On the other hand, the $(\min,+)$ matrix product which is at the heart of many fundamental graph problems such as All-Pairs Shortest Paths, has received only minor $n^{o(1)}$ improvements over its brute-force cubic running time and is widely conjectured to require $n^{3-o(1)}$ time. There is a plethora of matrix products and graph problems whose complexity seems to lie in the middle of these two problems. For instance, the Min-Max matrix product, the Minimum Witness matrix product, All-Pairs Shortest Paths in directed unweighted graphs and determining whether an edge-colored graph contains a monochromatic triangle, can all be solved in $\widetilde{O}(n^{(3+\omega)/2})$ time. While slight improvements are sometimes possible using rectangular matrix multiplication, if $\omega=2$, the best runtimes for these ``intermediate'' problems are all $\widetilde{O}(n^{2.5})$.

A similar phenomenon occurs for convolution problems. Here, using the FFT, the usual $(+,\times)$-convolution of two $n$-length sequences can be solved in $O(n\log n)$ time, while the $(\min,+)$-convolution is conjectured to require $n^{2-o(1)}$ time, the brute force running time for convolution problems. There are analogous intermediate problems that can be solved in $O(n^{1.5})$ time, but seemingly not much faster: Min-Max convolution, Minimum Witness convolution, etc.

Can one improve upon the running times for these intermediate problems, in either the matrix product or the convolution world? Or, alternatively, can one relate these problems to each other and to other key problems in a meaningful way?

This paper makes progress on these questions by providing a network of fine-grained reductions. We show for instance that APSP in directed unweighted graphs and Minimum Witness product can be reduced to both the Min-Max product and a variant of the monochromatic triangle problem, so that a significant improvement over $n^{(3+\omega)/2}$ time for any of the latter problems would result in a similar improvement for both of the former problems. We also show that a natural convolution variant of monochromatic triangle is fine-grained equivalent to the famous $3$SUM problem. As this variant is solvable in $O(n^{1.5})$ time and $3$SUM is in $O(n^2)$ time (and is conjectured to require $n^{2-o(1)}$ time), our result gives the first fine-grained equivalence between natural problems of different running times. We also relate $3$SUM to monochromatic triangle, and a coin change problem to monochromatic convolution, and thus to $3$SUM.
\end{abstract}

\section{Introduction}
Matrix multiplication is arguably the most fundamental linear algebraic operation. It is an important primitive for an enormous variety of applications. Within algorithmic research it has a very special role since it is one of the few problems for which we have surprisingly fast and completely counter-intuitive algorithms. Starting with Strassen's breakthrough~\cite{strassen} in 1969, a long line of research culminated in the current bound $\omega<2.373$~\cite{Vassilevska12,LeGall14}, where $\omega$ is the smallest real number so that $n\times n$ matrix multiplication can be performed in $\Oh{n^{\omega+\eps}}$ time for all $\eps>0$.

In many applications, one needs to compute matrix products that are a bit different (often called funny~\cite{AlonGM97}) from the usual definition of matrix multiplication over a ring such as the integers ($C_{ij}=\sum_k A_{ik} \cdot B_{kj}$). Such examples include matrix products over semirings such as the $(\min,+)$-product (often called distance product) which is over the tropical ($(\min,+)$) semiring, and the Max-Min product which is over the $(\max,\min)$-semiring.
Both these products are equivalent to certain types of path optimization problems in graphs. The distance product of $n\times n$ matrices is equivalent to the All-Pairs Shortest Paths (APSP) problem in $n$-node graphs, so that a $T(n)$ time algorithm for one problem would imply an $\Oh{T(n)}$ time algorithm for the other~\cite{fischermeyer}. Similarly, the Max-Min product is equivalent to the so called All-Pairs Bottleneck Paths (APBP) in graphs (e.g.~\cite{ShapiraYZ11}).

There seems to be a distinct complexity difference between APSP and APBP (and hence the corresponding matrix products), however. The fastest algorithms for APSP and the distance product run in $n^3/\exp{(\sqrt{\log n})}$ time~\cite{ryanapsp}, which is only better by an $n^{o(1)}$ factor than the trivial cubic time algorithm for the distance product. Meanwhile, as was first shown by~\cite{VassilevskaWY07,VassilevskaWY09}, APBP and the Max-Min product admit a much faster than cubic time algorithm via a reduction to (normal) matrix multiplication; the fastest running time is $\Oh{n^{(3+\omega)/2}}$~\cite{DuanP09}.

APSP is in fact conjectured to not admit any truly subcubic, $\Oh{n^{3-\eps}}$ time algorithms for $\eps>0$. Fine-grained complexity has strengthened this hypothesis by providing a large class of problems that are equivalent to APSP and the distance product, via fine-grained subcubic reductions. Thus the reason why distance product is seemingly so difficult is because there are many problems that are equivalent to it and researchers from different communities have all failed to solve these problems faster.

The best known running time for the $n\times n$ Max-Min product, $\Ot{n^{(3+\omega)/2}}$, while nontrivially subcubic, seems difficult to improve upon. In fact,  $\Ot{n^{(3+\omega)/2}}$ is the best known running time for many other matrix and graph problems besides the Max-Min product: 
the Dominance product~\cite{Matousek91} and Equality product~\cite{equalityhw,Labib19}, 
All-Pairs Nondecreasing Paths (APNP) and the $(\min,\leq)$-product~\cite{vnondec,vnondecj,DuanJW19}.
For some of these problems~\cite{yusterdom,Gold17} one can obtain slightly improved running times using rectangular matrix multiplication~\cite{legallurrutia}. However, the closer $\omega$ is to $2$, the smaller the improvements, and when $\omega=2$, the $\Ot{n^{(3+\omega)/2}}=\Ot{n^{2.5}}$ running time is the best known for all of these problems.
Since their running time exponent is essentially the average of the brute-force exponent $3$ and the fast matrix multiplication exponent $\omega$, we will call these problems ``intermediate''.

Next two problems that are intermediate if $\omega=2$ are: the Minimum Witness product, which is related to the problem of computing All-Pairs Least Common Ancestors in a DAG, and All-Pairs Shortest Paths (APSP) in unweighted directed graphs.
For both problems we know algorithms running in $\Ot{n^{(3+\omega)/2}} \leq \Ot{n^{2.687}}$ time~\cite{BenderPSS01,AlonGM97}, and both algorithms can be improved upon, by using rectangular matrix multiplication~\cite{Czumaj07LCA,Zwick02}. The improvement is already seen in a naive implementation, i.e.~cutting rectangular matrices into square blocks, which gives an $\Ot{n^{2+1/(4-\omega)}} \leq \Ot{n^{2.615}}$~time. Employing a specialized rectangular matrix multiplication algorithm~\cite{legallurrutia}, brings the runtime down to $\Ot{n^{2.529}}$. When $\omega=2$, however, all the improvements vanish and those running times become~$\Ot{n^{2.5}}$.

\begin{center}\emph{Is the $2.5$ running time exponent (for $\omega=2$) for all of these problems a coincidence, or can we relate all of them via fine-grained reductions, and use plausible hypotheses to explain it?}\end{center}
This is a question that many have asked, but unfortunately there are only two partial answers: First, it is known that Equality product and Dominance product are equivalent (\cite{equalityhw,Labib19}, also follows from Proposition~3.4 in~\cite{WilliamsW13}), and that they are equivalent to All-Pairs $\ell_{2p+1}$ Distances~\cite{Labib19}. The second result is that the Max-Min product is equivalent to approximate APSP in weighted graphs without scaling~\cite{Bringmann19Apx}. The main question above remains \emph{wide open}.

Parallel to the world of matrix products, there is a very similar landscape of \emph{convolution} problems. While it is well-known that the $(+,\times)$-convolution\footnote{The $(+,\times)$-convolution of two vectors $a$ and $b$ is the vector $c$ such that $c_i = \sum_j a_jb_{i-j}$.} of two $n$-length vectors can be computed in $\Oh{n \log n}$ time using the Fast Fourier Transform (FFT), these techniques no longer work for the $(\min,+)$-convolution, and this problem is conjectured to require $n^{2-o(1)}$ time (see e.g.~\cite{Cygan19}). Similar to the ``intermediate'' matrix product problems, there are analogous ``intermediate'' convolution problems, all in $\Ot{n^{3/2}}$ time\footnote{The exponent $(3+\omega)/2$ for intermediate matrix products is the average of the fast matrix multiplication exponent and the brute force matrix product exponent, and the exponent $3/2$ for intermediate convolution problems is the average of the fast convolution exponent $1$ and the brute force exponent $2$.}: Max-Min convolution, Dominance convolution, Minimum Witness convolution, etc. 
 
The convolution landscape is even somewhat cleaner than the matrix product one. As the normal convolution ($(+,\times)$) is already in (near-)linear time, there are no analogues of rectangular matrix multiplication speedups, and all intermediate problems happen to have exactly the same running time (up to polylogarithmic factors). Still, there is no real formal explanation of why they have the same running time. The only reductions between these convolutions are analogous to the matrix product ones: Dominance convolution is equivalent to Equality convolution~\cite{Labib19}, and approximate $(\min,+)$-convolution is equivalent to exact Max-Min convolution~\cite{Bringmann19Apx}.

\subsection{Our contributions}
In this paper we provide new fine-grained reductions between  several intermediate matrix product and all-pairs graph problems, and between intermediate convolution problems, also relating these to other key problems from fine-grained complexity such as $3$SUM. See Figure~\ref{fig:reductions} for a pictorial representation of our results.

\begin{figure}[t]
\centering
\begin{tikzpicture}[every edge/.style={draw,thick}]
  \node (monoconv) at (0,2) {\monoconv};
  \node (3sum)   at (3,2) {$3$SUM};
  \node (sparse) at (6,-0.1) [label={[name=listing]below:{\textsf{TriangleListing} $(t=m)$}}] {\sparse};
  \node (mono)   at (6,2) {\mono};
  \node (minw)   at (9,4) {\minw};
  \node (minmax) at (9,2) [label={[align=center,name=apbp]below:{\textsf{AP-BottleneckPaths}\\\textsf{ApproximateAPSP}}}] {\minmax};
  \node (apsp)   at (6,4) {\uapsp};
  \node (coin)   at (0,4) {\coin};
  
  \path [->] (coin) edge [dotted] (monoconv);
  \path [->] (mono) edge (sparse);

  \path [->] (3sum) edge [dashed] (mono);
 
  \path [->] (minw) edge [dashed] (mono);
  \path [->] (apsp) edge [dashed] (mono);
  \path [->] (minw) edge [dashed] (minmax);
  \path [->] (apsp) edge [dashed] (minmax);
  
  \path [->] (monoconv) edge[bend left=20] (3sum);
  \path [->] (3sum) edge[bend left=20] (monoconv);
  
  \node [draw=gray, dashed, rounded corners, inner sep=4pt, fit=(monoconv), label={[gray, label distance=6pt]below:{$\Ot{n^{1.5}}$}}] {};
  \node [draw=gray, dashed, rounded corners, inner sep=4pt, fit=(coin), label={[gray, label distance=4pt]right:{$\Ot{n^{4/3}}$}}] {};
  \node [draw=gray, dashed, rounded corners, inner sep=4pt, fit=(3sum), label={[gray, label distance=6pt]below:{$\Oh{n^2}$}}] {};
  \node [draw=gray, dashed, rounded corners, inner sep=4pt, fit=(sparse) (listing), label={[gray,label distance=4pt]left:{$\Oh{m^{2\omega/(\omega+1)}}$}}] {};
  \node [draw=gray, dashed, rounded corners, inner sep=4pt, fit=(apsp) (mono) (minw) (minmax) (apbp), label={[gray,label distance=4pt, xshift=65pt]below:{$\Ot{n^{(3+\omega)/2}}$}}] {};
\end{tikzpicture}
\caption{Our results. An arrow pointing from problem $A$ to problem $B$ means that problem $A$~reduces to problem $B$ in the fine-grained sense.
Solid arrows denote reductions which are tight with respect to the best currently known running times, i.e.~improving by a polynomial factor over the best known running time for one problem implies a polynomial improvement over the best known running time for the other.
Dashed arrows denote reductions which become tight when $\omega=2$.
The reduction from \coin to \monoconv, denoted by a dotted arrow, is not tight.}
\label{fig:reductions}
\end{figure}
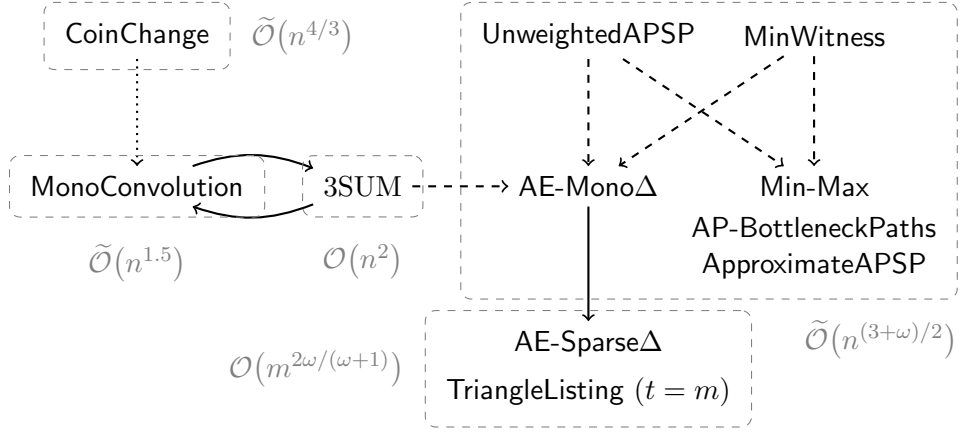

\paragraph*{Reductions for Graph Problems and Matrix Products.}
Several of our reductions concern the All-Edges Monochromatic Triangle (\mono) problem: Given an $n$-node graph in which each edge has a color from $1$ to $n^2$, decide for each edge whether it belongs to a \emph{monochromatic} triangle, a triangle whose all three edges have the same color. Vassilevska, Williams and Yuster~\cite{Hsubgraphs} studied the decision variant of \mono in which one asks whether the given graph contains a monochromatic triangle. They provided an $\Oh{n^{(3+\omega)/2}}$ time algorithm for the decision problem, but that algorithm is in fact strong enough to also solve the all-edges variant \mono, making \mono one of the ``intermediate'' problems of interest.

To obtain their $\Oh{n^{(3+\omega)/2}}$ time algorithm, Vassilevska, Williams and Yuster~\cite{Hsubgraphs} implicitly reduce \mono (in a black-box way) to the \sparse problem of deciding for every edge $e$ in an $m$-edge graph whether $e$ is in a triangle. The fastest known algorithm for \sparse is by Alon, Yuster and Zwick~\cite{Alon97}, running in $\Oh{m^{2\omega/(\omega+1)}}$ time, and the problem is known to be runtime equivalent to the problem of \emph{listing} up to $m$ triangles in an $m$-edge graph~\cite{Duraj19}. The black-box reduction of~\cite{Hsubgraphs} from \mono to \sparse implies that a significant improvement over the $\Oh{m^{2\omega/(\omega+1)}}$ time for \sparse would translate to an improvement over $\Oh{n^{(3+\omega)/2}}$ for \mono.

\begin{restatable}[implicit in~\cite{Hsubgraphs}]{theorem}{ThmMonoToSparse}
\label{thm:MonoToSparse}
If \sparse is in $\Oh{m^{2\omega/(\omega+1) - \eps}}$ time, for some $\eps > 0$, then \mono is in $\Oh{n^{(3+\omega)/2 - \delta}}$ time, for some $\delta > 0$.
\end{restatable}

Our first set of results shows that \mono is powerful enough to capture two well-studied intermediate problems: the Minimum Witness product of two Boolean matrices and the All-Pairs Shortest Paths problem in directed unweighted graphs.

The Minimum Witness product (\minw) $C$ of two Boolean matrices $A$ and $B$ is defined as $C_{ij}=\min \{k \mid A_{ik}=B_{kj}=1\}$ (where the minimum is defined to be $\infty$ if there is no witness~$k$). \minw is used, e.g., for determining for every pair $u,v$ of vertices in a DAG, the least common ancestor of $u$ and $v$, i.e. solving the All-Pairs Least Common Ancestors problem~\cite{Czumaj07LCA}. The fastest known algorithm for \minw runs in $\Oh{n^{2.529}}$ time using rectangular matrix multiplication, and in $\Oh{n^{2+1/(4-\omega)}}$ time just using square matrix multiplication~\cite{Czumaj07LCA}.

The All-Pairs Shortest Paths (APSP) problem in unweighted graphs is very well-studied. While in undirected graphs, the problem is known to be solvable in $\Ot{n^\omega}$ time~\cite{seidel}, the problem in directed graphs is one of our intermediate problems. Its fastest algorithm (similarly to \minw) runs in $\Oh{n^{2.529}}$ time using rectangular matrix multiplication, and in $\Ot{n^{2+1/(4-\omega)}}$ time just using square matrix multiplication~\cite{Zwick02}. We will refer to the APSP problem in directed unweighted graphs as \uapsp.

We present reductions from \minw and \uapsp to \mono with only polylogarithmic overhead. 

\begin{restatable}{theorem}{ThmMinwToMono}
\label{thm:MinwToMono}
If \mono is in $T(n)$ time, then \minw is in $\Oh{T(n)\log n}$ time.
\end{restatable}

\begin{restatable}{theorem}{ThmApspToMono}
\label{thm:ApspToMono}
If \mono is in $T(n)$ time, then \uapsp is in $\Oh{T(n)\log^2 n}$ time.
\end{restatable}

The above reductions tightly relate \minw and \uapsp to \mono if $\omega=2$, showing that any improvement over the $2.5$ exponent for \mono, gives the same improvement for \minw and \uapsp. Due to the tight reduction (Theorem~\ref{thm:MonoToSparse}) from \mono to \sparse, we also obtain that an $\Oh{m^{4/3-\eps}}$ time algorithm, with $\eps>0$, for \sparse would give $\Oh{n^{2.5-\delta}}$ time algorithms, for $\delta>0$, for \minw and \uapsp, presenting another tight relationship for the case when $\omega=2$.

Our next result is that improving over the exponent $2.5$ for \mono is at least as hard as obtaining a truly subquadratic time algorithm for the $3$SUM problem.

\begin{restatable}{theorem}{ThmTsumToMono}
\label{thm:3sumToMono}
If \mono is in $\Oh{n^{5/2-\eps}}$ time, then $3$SUM is in (randomized) $\Ot{n^{2-\frac{4}{5}\eps}}$ time.
\end{restatable}

In $3$SUM one is given $n$ integers and is asked whether three of them sum to $0$. The problem is easy to solve in $\Oh{n^2}$ time, and slightly subquadratic time algorithms exist~\cite{Baran08,Chan18}. $3$SUM is a central problem in fine-grained complexity~\cite{vsurvey}. It is hypothesized to require $n^{2-o(1)}$ time (on a word-RAM with $\Oh{\log n}$ bit words), and many fine-grained hardness results are conditioned on this hypothesis (see~\cite{GajentaanO95,vsurvey}). Our reduction shows that, under the $3$SUM Hypothesis, the exponent $2.5$ for \mono cannot be beaten, and this is tight if $\omega=2$.
We note that before our work no intermediate matrix, graph, or convolution problem was known to be $3$SUM-hard.

Next, we consider the Min-Max product (\minmax) of two matrices $A$ and $B$, defined as $C_{ij}=\min_k \max(A_{ik},B_{kj})$. The Min-Max product is equivalent to the aforementioned Max-Min product (just negate the matrix entries) and the All-Pairs Bottleneck Paths problem, and is thus solvable in $\Oh{n^{(3+\omega)/2}}$ time~\cite{DuanP09}.

A very simple folklore reduction shows that \minmax on $n\times n$ integer matrices is at least as hard as \minw on $n\times n$ Boolean matrices, giving a tight relationship when $\omega=2$.

\begin{restatable}[folklore]{theorem}{ThmMinwToMinmax}
\label{thm:MinwToMinmax}
If \minmax is in $T(n)$ time, then \minw is in $\Oh{T(n)}$ time.
\end{restatable}

Our next result states that the All-Pairs Shortest Paths problem in directed unweighted graphs (\uapsp) is also tightly reducible to \minmax. This gives a second intermediate problem that is at least as hard as both \minw and \uapsp.

\begin{restatable}{theorem}{ThmApspToMinmax}
\label{thm:ApspToMinmax}
If \minmax is in $T(n)$ time, then \uapsp is in $\Oh{T(n)\log n}$ time.
\end{restatable}

The above theorem also follows from a recent independent result by Barr, Kopelowitz, Porat and Roditty~\cite{Barr19}. In particular, they reduce All-Pairs Shortest Paths in directed graphs with edge weights from $\{-1,0,1\}$ to \minmax. Interestingly, they use a substantially different approach than ours. While their argument can be seen as inspired by Seidel's algorithm for unweighted APSP in undirected graphs~\cite{seidel}, ours resembles Zwick's algorithm for directed graphs~\cite{Zwick02}.

\paragraph*{Reductions for Convolution Problems.}
Our main result for convolution problems regards the convolution version of $\mono$, which we call \monoconv: Given three integer sequences $a,b,c$, decide for each index $i$ if there exists $j$ such that $a_j = b_{i-j} = c_i$. We show that \monoconv is actually fine-grained equivalent to $3$SUM. 

\begin{restatable}{theorem}{ThmTsumToConv}
\label{thm:3sumToConv}
If \monoconv is in $\Oh{n^{3/2-\eps}}$ time, then $3$SUM is in (randomized) $\Ot{n^{2-\frac{4}{3}\eps}}$ time.
\end{restatable}

\begin{restatable}{theorem}{ThmConvToTsum}
\label{thm:ConvTo3sum}
If $3$SUM is in $\Oh{n^{2-\eps}}$ time, then \monoconv is in $\Ot{n^{3/2-\eps/(8-2\eps)}}$ time.
\end{restatable}

This equivalence is arguably the first fine-grained equivalence between natural problems with \emph{different} running time complexities: $\monoconv$ is a problem in $\Oh{n^{3/2}}$ time, whereas $3$SUM is in $\Oh{n^2}$ time, and a polynomial improvement on one of these running times would result in a polynomial improvement over the other. All previous fine-grained equivalences were between problems with the same running time exponent: the problems equivalent to APSP~\cite{Subcubic,AbboudGW15} are all solvable in $\Oh{N^{1.5}}$ time where $N$ is the size of their input, the problems equivalent to Orthogonal Vectors~\cite{ChenW19} or to $(\min,+)$-convolution~\cite{Cygan19} are all in quadratic time, the problems equivalent to CNF-SAT~\cite{CyganDLMNOPSW16} are all in $\Oh{2^n}$ time, etc. 
While tight fine-grained reductions between problems with different running times are well-known, there was no such equivalence until our result, largely since it often seems difficult to reduce a problem with a smaller asymptotic running time to one with a larger running time, something our Theorem~\ref{thm:ConvTo3sum} overcomes.
Note that the same apparent difficulty is overcome by the reduction from \mono to \sparse in Theorem~\ref{thm:MonoToSparse}, as well as by the reductions from \minw and \uapsp to \sparse, which follow from combining Theorems~\ref{thm:MinwToMono} and~\ref{thm:ApspToMono} with  Theorem~\ref{thm:MonoToSparse}.

Theorem~\ref{thm:ConvTo3sum} together with Theorem~\ref{thm:3sumToMono} give a reduction from \monoconv to \mono. Previously reductions from a convolution to the corresponding graph/matrix problem were known only for problems with best known algorithms running in brute-force time, i.e.~quadratic time for convolution and cubic time for product, e.g.~$(\min,+)$-convolution reduces to $(\min,+)$-product~\cite{convolutionMinPlus}.

Finally, we relate \monoconv to an unweighted variant of a coin change problem~\cite{Wright1975,Kunnemann17} that is related to the minimum word break problem~\cite{Bringmann17,Chan20}. Given a set of coin values from $\{1,2,\ldots,n\}$, the \coin problem asks to determine for each integer value up to $n$ what is the minimum number of coins (allowing repetitions) that sum to that value. We reduce \coin to \monoconv with only a polylogarithmic overhead.
A simple algorithm solves \coin in $\Ot{n^{3/2}}$ time~\cite{Bringmann19Coin}, and our reduction implies that any improvement over the known running times of \monoconv or $3$SUM would also improve over the above running time for \coin. Following the publication of the conference version of this paper, Chan and He~\cite{Chan20} gave a faster $\Ot{n^{4/3}}$ time algorithm for \coin. Therefore, our reduction is no longer tight with respect to the best currently known running times. In order to improve over Chan and He's running time using our reduction one would need an $\Oh{n^{4/3-\eps}}$ time algorithm for \monoconv.

\begin{restatable}{theorem}{ThmCoinToConv}
\label{thm:CoinToConv}
If \monoconv is in $T(n)$ time, then \coin is in $\Oh{T(n)\log^2 n}$ time.
\end{restatable}

\section{Preliminaries}

In this section we first recall formal definitions of all the problems involved in the reductions presented in the paper. We split these problems by their time complexity. At the end of the section we recall the property of self-reducibility of $3$SUM.

\subsection{Problems in \texorpdfstring{$\Ot{n^{(3+\omega)/2}}$}{O(n\^{}((3+omega)/2))} time}

\begin{definition}[All-Edges Monochromatic Triangle, \mono]
Given an $n$-node graph $G$ in which each edge has a color from $1$ to $n^2$, decide for each edge whether it belongs to a \emph{monochromatic} triangle, a triangle where all three edges have the same color. 
\end{definition}

\begin{definition}[Min-Max matrix product, \minmax]
Given two $n \times n$ matrices $A$ and $B$, compute matrix $C$ such that
\[C_{ij} = \min_k \max(A_{ik}, B_{kj}).\]
\end{definition}

\begin{definition}[Minimum Witness matrix product, \minw]
Given two $n \times n$ Boolean matrices $A$ and $B$, compute matrix $C$ such that
\[C_{ij} = \min (\{k \mid A_{ik} = B_{kj} = 1\} \cup \{\infty\}).\]
\end{definition}

\begin{definition}[All-Pairs Shortest Paths in directed unweighted graphs, \uapsp] 
Given an $n$-node unweighted directed graph $G=(V,E)$, compute for each pair of vertices $u, v \in V$ the length of a shortest path from $u$ to $v$.
Note that all path lengths will be in $\{0,1,\ldots,n-1\}\cup \{\infty\}$.
\end{definition}

\subsection{Problems in \texorpdfstring{$\Oh{m^{2\omega/(\omega+1)}}$}{O(m\^{}(2omega/(omega+1)))} time}

\begin{definition}[All-Edges Sparse Triangle, \sparse]
Given an $m$-edge graph $G$ decide for each edge whether it belongs to a triangle.
\end{definition}

\subsection{Problems in \texorpdfstring{$\Oh{n^2}$}{quadratic} time}

\begin{definition}[$3$SUM]
Given three lists, $A$, $B$ and $C$, of $n$ integers, determine if there exist $a\in A$, $b \in B$, and $c\in C$ such that $a+b=c$.
\end{definition}

Let us note that the $3$SUM problem is defined in several different ways in literature. They differ as to whether the input is split into three list or all the numbers are in a single list, and whether one looks for $a+b=c$ or $a+b+c=0$. All these variants are equivalent by simple folklore reductions.

\subsection{Problems in \texorpdfstring{$\Ot{n^{1.5}}$}{O(n\^{}1.5)} time}

\begin{definition}[\monoconv]
Given three sequences $a,b,c$, all of length $n$, compute the sequence $d$ such that
\[d_i = \begin{cases} 1 &\text{if } \exists_j\ a_j = b_{i-j} = c_i,\\
0 & \text{otherwise}.\end{cases}\]
\end{definition}

\subsection{Problems in \texorpdfstring{$\Ot{n^{4/3}}$}{O(n\^{}(4/3))} time}

\begin{definition}[\coin]
Given a set of coin values $C \subseteq \{1,2, \ldots,n\}$, assume you have for each $c \in C$ an infinite supply of coins of value $c$, and determine for each $v \in   \{1,2, \ldots,n\}$ the minimum number of coins that sums up to $v$.
\end{definition}

\coin can be easily solved in $\Ot{n^{1.5}}$ time~\cite{Bringmann19Coin}.
The algorithm splits the coins into heavy coins, with weight at least $\sqrt{n}$, and light coins, with weight less than $\sqrt{n}$. The minimum sum for a value can use at most $\sqrt{n}$ heavy coins. By running FFT $\sqrt{n}$ times the algorithm produces a vector with the minimum number of heavy coins needed to sum to every value. That takes $\Oh{n^{1.5} \log n}$ time in total. 
Then a classical dynamic programming algorithm is run for the $\sqrt{n}$ light coins and $n$ values, in $\Oh{n^{1.5}}$ time. 

For a more involved $\Ot{n^{4/3}}$ time algorithm refer to~\cite{Chan20}.

\subsection{Self-reducibility of 3SUM}

In our proofs of Theorems~\ref{thm:3sumToMono} and~\ref{thm:3sumToConv} we use the following fact about $3$SUM.

\begin{lemma}
\label{lem:3sumself}
For any $\alpha \in [0, 1]$, a single instance of $3$SUM of size $n$ can be reduced to 
$\Oh{n^{2\alpha}}$ instances of $3$SUM of size $\Oh{n^{1-\alpha}}$ each. The reduction runs in time linear in the total size of produced instances, and the original instance is a yes-instance if and only if at least one of the produced instances is a yes-instance.
\end{lemma}

This fact appears in many $3$SUM-related papers, e.g.~\cite{Baran08,Hajiaghayi19,Kopelowitz16,Lincoln16,Patrascu10}. In~\cite{Baran08} it was proved using a randomized almost linear hashing scheme~\cite{Dietzfelbinger96}. An alternative proof -- using a domination argument to provide a deterministic reduction -- appeared, e.g., in~\cite{Lincoln16,Gronlund18}, and is based on ideas of~\cite{Czumaj07Tri}.

\section{Reductions for Graph and Matrix Problems}

First, let us recall the algorithm of Vassilevska, Williams and Yuster~\cite{Hsubgraphs} for \mono. We rephrase the argument so that it not only shows how to solve \mono in $\Oh{n^{(3+\omega)/2}}$ time, but also proves that any polynomial improvement over the $\Oh{m^{2\omega/(\omega+1)}}$ time algorithm of Alon, Yuster and Zwick~\cite{Alon97} for \sparse translates to a polynomial improvement for \mono.

\ThmMonoToSparse*
\begin{proof}
Assume \sparse is in $\Oh{m^\alpha}$ time. Take an \mono instance. For each color consider the subgraph composed of all the edges of that color. Each such subgraph constitutes an independent instance of \sparse. However, simply using the $\Oh{m^\alpha}$ time algorithm on all of these instances is not efficient enough. Intuitively, some of the instances might be too dense.

Instead, for a parameter $t$ to be determined later, take the $t$ largest subgraphs (in terms of the number of edges). For each of them solve the problem by using fast matrix multiplication to compute the square of the adjacency matrix. This takes $\Oh{tn^\omega}$ time in total. Let $m_i$ denote the number of edges in the $i$-th of the remaining subgraphs. Clearly, $\forall_i\ m_i \leq n^2/t$, and $\sum_i m_i \leq n^2$. On each of those subgraphs use the $\Oh{m^\alpha}$ time \sparse algorithm. This takes an order of
\[\sum_i m_i^\alpha = \sum_i m_i\cdot m_i^{\alpha-1} \leq \sum_i m_i\cdot(n^2/t)^{\alpha-1} \leq n^2 \cdot(n^2/t)^{\alpha-1}\]
time. The total runtime is thus $\Oh{tn^\omega + n^{2\alpha}/t^{\alpha-1}}$. Optimize by setting $t=n^{(2\alpha-\omega)/\alpha}$, and get an $\Oh{n^{\omega + 2 - (\omega / \alpha)}}$ time.

Observe that for $\alpha = 2\omega/(\omega + 1)$ the runtime is $\Oh{n^{(3+\omega)/2}}$. Moreover, for $\alpha < 2\omega/(\omega + 1)$ the exponent in the runtime becomes strictly smaller.
\end{proof}

Now, we proceed to show how to use \mono to solve two popular intermediate problems. We start with \minw, and reduce a single instance of that problem to $\log n$ instances of \mono.

\ThmMinwToMono*
\begin{proof}
The main idea is to use a parallel binary search. For each entry of the output matrix $C$ we will keep an interval which that entry is guaranteed to lie in. With a single call to \mono we will be able to halve all the intervals.

W.l.o.g.~assume the last column of $A$ and last row of $B$ are all ones, so that the output is always finite. For $\ell\in[\log n]$, let $C^{(\ell)}$ denote the matrix pointing to $2^\ell$-length intervals in which entries of $C$ lie, that is $C^{(\ell)}_{ij}$ is the unique integer such that $C_{ij} \in \big[2^\ell \cdot C^{(\ell)}_{ij}, 2^\ell \cdot (C^{(\ell)}_{ij} + 1)\big)$.

We will compute $C^{(\ell)}$ for $\ell=\lceil \log n \rceil,\ldots,1,0$. Observe that $C^{(\lceil \log n \rceil)}$ is the zero matrix. Knowing $C^{(\ell+1)}$, we compute $C^{(\ell)}$ as follows. We create a tripartite graph $G=(I \cup J \cup K, E)$, with each of $I,J,K$ containing $n$ vertices. We add edges between $I$ and $K$ according to the matrix A. Edges from the $k$-th column get the label $\lfloor k / 2^\ell \rfloor$. We add edges between $K$ and $J$ according to the matrix B. Edges from the $k$-th row get the label $\lfloor k / 2^\ell \rfloor$. Finally, we add the full bipartite clique between $I$ and $J$. The edge between the $i$-th vertex of $I$ and the $j$-th vertex of $J$ gets the label $2 \cdot C^{(\ell+1)}$. That edge forms a monochromatic triangle if and only if $C_{ij} \in \big[ 2^\ell \cdot 2 \cdot C^{(\ell+1)}_{ij},
      2^\ell \cdot(2 \cdot C^{(\ell+1)}_{ij} + 1) \big)$,
i.e.~$C^{(\ell)}_{ij}=2 \cdot C^{(\ell + 1)}_{ij}$. Otherwise, it must be that $C_{ij} \in \big[ 2^\ell \cdot (2 \cdot C^{(\ell+1)}_{ij} + 1),
      2^\ell \cdot (2 \cdot C^{(\ell+1)}_{ij} + 2) \big)$,
i.e.~$C^{(\ell)}_{ij}=2 \cdot C^{(\ell + 1)}_{ij} + 1$. Therefore, solving \mono on $G$ suffices to compute $C^{(\ell)}$. Finally, observe that $C=C^{(0)}.$
\end{proof}

With a slightly more involved argument we show how to solve \uapsp with $\Oh{\log^2 n}$ calls to \mono.

\ThmApspToMono*
\begin{proof}

We solve \uapsp in $\log n$ rounds, in the $i$-th round we compute matrix $D^{\leq 2^i}$ of lengths of shortest paths of length up to $2^i$ (other entries equal to $\infty$). Each round will consist of a parallel binary search, similar to the one we use in our reduction from \minw to \mono (Theorem~\ref{thm:MinwToMono}). The algorithm is based on the fact that in unweighted graphs every path can be split roughly in half, i.e.~if the distance from $u$ to $v$ equals to $k$, then there must exist a vertex $w$ such that the distances from $u$ to $w$ and from $w$ to $v$ equal to $\lfloor k/2 \rfloor + \{0,1\}$.

To start, note that $D^{\leq 2^0}$ is a $\{0,1,\infty\}$-matrix that can be easily obtained from the adjacency matrix of the input graph. Now, assume we already computed $D^{\leq 2^i}$ and let us proceed to compute $D^{\leq 2^{i+1}}$. To avoid excessive indexing, let $A$ denote $D^{\leq 2^i}$, and $B$ denote $D^{\leq 2^{i+1}}$. For each entry of the output matrix $B$ we will keep an interval which that entry is guaranteed to lie in. With a single call to \mono we will be able to halve all the intervals.

For $\ell\in\{0,1,\ldots,i+2\}$, let $B^{(\ell)}$ denote the matrix pointing to $2^\ell$-length intervals in which entries of $B$ lie, that is $B^{(\ell)}_{uv}$ equals to the unique integer such that $B_{uv} \in \big[2^\ell \cdot B^{(\ell)}_{uv}, 2^\ell \cdot (B^{(\ell)}_{uv} + 1) - 1\big]$, or to infinity in case $B_{uv}$ is infinite.

We will iterate over $\ell$ from $i+2$ down to $0$. First, we need to compute $B^{(i+2)}$, whose entries are either zeros or infinities. Recall that we already know the matrix $A=D^{\leq 2^i}$. Consider a pair of nodes $u$ and $v$ that are at distance at most $2^{i+1}$. There must exist a node $w$ such that $A_{uw} \leq 2^i$ and $A_{wv} \leq 2^{i}$, that is, equivalently both $A_{uw}$ and $A_{wv}$ are finite. We obtain the matrix $B^{(i+2)}$ by squaring the $(0,1)$ matrix obtained from $A$ by putting ones at the finite entries and zeros elsewhere. That single Boolean matrix multiplication can be easily simulated by a single call to \mono, using just two colors.

Once we have the matrix $B^{(\ell+1)}$ we want to compute $B^{(\ell)}$. For this we first note that if $B^{(\ell+1)}_{uv}=j$ then $B^{(\ell)}_{uv}$ is either $2j$ or $2j+1$.
If $B^{(\ell)}_{uv}=2j$, then there must exist a vertex $w$ such that
\begin{equation}
\label{eqn:apsp}    
A_{uw}\in \big[2^{\ell-1} \cdot (2j), 2^{\ell-1} \cdot (2j+1)\big), \quad\text{and}\quad   A_{wv}\in \big[2^{\ell-1} \cdot (2j), 2^{\ell-1} \cdot (2j+1)\big].
\end{equation}
Furthermore, if $B^{(\ell)}_{uv}>2j$, then there is no $w$ such that the above condition holds. This will allow us to distinguish between the $2j$ and $2j+1$ cases by coloring the matrix $A$ based on which range the entries fall in. Note that the ranges in Condition~(\ref{eqn:apsp}) do not overlap with corresponding ranges for different integer values $j' \ne j$. Thus we will be able to use a single call to \mono to check in parallel for all values of $B^{(\ell)}_{uv}$ if they are the smaller even value $2\cdot B^{(\ell+1)}_{uv}$ or the larger odd value $2\cdot B^{(\ell+1)}_{uv}+1$.

We construct an \mono~instance with a tripartite graph  with the vertex set $U \sqcup V \sqcup W$ where $U$, $V$ and $W$ are disjoint copies of the original vertex set.
The edges between $U$ and $V$ correspond to our desired output. If $B^{(\ell+1)}_{uv} =j$ then we color the edge $(u,v) \in U \times V$ with $j$. 
The edges between $U$ and $W$ correspond to the first part of Condition~(\ref{eqn:apsp}), i.e.~if $A_{uw} \in \big[2^{\ell-1}\cdot (2j), 2^{\ell-1}\cdot (2j+1)\big)$, then we add the edge $(u,w)$ in $U\times W$ with color $j$.
The edges between $W$ and $V$ correspond to the second part of Condition~(\ref{eqn:apsp}), i.e.~if $A_{wv} \in \big[2^{\ell-1}\cdot (2j), 2^{\ell-1} \cdot (2j+1)\big]$, then we add the edge $(w,v)$ in $W \times V$ with color $j$.
Any edge $(u,v)$ in $U \times V$ that is in a monochromatic triangle implies $B^{(\ell)}_{uv} =2\cdot B^{(\ell+1)}_{uv}$. Conversely, any edge $(u,v)$ that is not a part of any monochromatic triangle implies $B^{(\ell)}_{uv} =2 \cdot B^{(\ell+1)}_{uv}+1$.

We iterate down until $B^{(0)}$, and observe that $B^{(0)} = B$. Thus, with $\Oh{\log n}$ calls we can compute $B= D^{\leq 2^{i+1}}$ from $A = D^{\leq 2^i}$. To solve \uapsp the total number of calls we need to make to \mono~is $\Oh{\log^2(n)}$.
Therefore, if \mono~can be solved in $T(n)$ time, then \uapsp~can be solved in $\Oh{T(n) \log^2(n)}$ time.
\end{proof}

Now we show that \mono is $3$SUM-hard. In our proof we use as a black-box the following reduction from $3$SUM to \sparse.

\begin{lemma}[Kopelowitz, Pettie, Porat~\cite{Kopelowitz16}]
\label{lem:sparse}
A single instance of $3$SUM of size $n$ can be reduced to a single instance of \sparse with $\Th{n \log n}$ vertices and $\Th{n^{3/2} \log n}$ edges.
\end{lemma}

\ThmTsumToMono*
\begin{proof}
Given an instance of $3$SUM of size $N$, we use the self-reduction (Lemma~\ref{lem:3sumself}), and reduce it to $\Oh{N^{2/5}}$ instances of size $\Oh{N^{4/5}}$ each. Then, we reduce each of these instances to an \sparse instance with $n=\Th{N^{4/5} \log N}$ vertices and $m=\Th{N^{6/5} \log N}$ edges, using Lemma~\ref{lem:sparse}. Now we will show how to combine these $\Oh{N^{2/5}}$ \sparse instances to form polylogarithmically many \mono instances, each with $\Oh{N^{4/5} \log N}$ vertices, which will finish the proof.

Assume w.l.o.g.~that all the created graphs are over the same vertex set $[n]$.
If we were lucky enough and the edge sets of the created \sparse instances were disjoint, the reduction would be essentially done. Indeed, we could simply union the edge sets to create a single graph, and use colors to track from which graph every edge originates. Solving that one \mono instance would provide answers to all \sparse instances. Sadly, the chances of such a favorable collision-free scenario are very slim. The remaining part of the proof shows how to deal with multiple \sparse instances containing the same edge.

We randomly permute the vertex sets, for each graph independently. For a fixed $(u,v) \in [n]^2$, such that $u \neq v$, the probability that a fixed graph contains the edge $(u,v)$ equals to $p = m / \binom{n}{2} = \Oh{(N^{2/5} \log N)^{-1}}$. The expected number of $(u,v)$ edges across all graphs is $\Oh{N^{2/5} \cdot p} = \Oh{1/\log N}$. By a Chernoff bound, the probability that the number of $(u,v)$ edges exceeds $c \log n$ is less than $(1/e)^{\Th{c \log n}}$. We take $c$ large enough so that, by the union bound over all possible $\binom{n}{2}$ edges, with probability at least $1/2$ no edge appears more than $c \log n$ times across all graphs. For each $(u,v) \in [n]^2$ we arbitrarily number all $(u,v)$ edges with consecutive positive integers from $1$ up to at most $c \log n$. We iterate over all triples $(i,j,k) \in [c \log n]^3$. For every triple we create a tripartite graph with the vertex set $V_1 \sqcup V_2 \sqcup V_3$, for $V_1=V_2=V_3=[n]$. We create an edge $(u,v)$ between $V_1$ and $V_2$ if there exists an edge $(u,v)$ with number $i$ assigned to it in any of the \sparse instances. Note that there is at most one such instance. We set the color of the newly created edge to the identifier of the instance it originates from. Similarly, we create edges between $V_2$ and $V_3$ using edges with number $j$ assigned, and between $V_3$ and $V_1$ using number $k$. That gives us $(c \log n)^3$ instances of \mono. Note that every triangle present in any of the \sparse instance corresponds to a single monochromatic in one of the \mono instances, and vice versa. We solve all \mono instances and combine the outputs in order to get the output for all \sparse instances, and eventually for the $3$SUM instance.
\end{proof}

The next two theorems use techniques similar to Theorems~\ref{thm:MinwToMono} and~\ref{thm:ApspToMono} to give reductions to \minmax.

\ThmMinwToMinmax*
\begin{proof}
Given two $(0,1)$ matrices $A$ and $B$, we construct matrices $A'$ and $B'$ such that
\[A'_{ik} = \begin{cases}k &\text{if } A_{ik} = 1,\\\infty &\text{if } A_{ik} = 0,\end{cases} \quad \text{and} \quad B'_{kj} = \begin{cases}k &\text{if } B_{kj} = 1,\\\infty &\text{if } B_{kj}=0.\end{cases}\]
Observe that the $(\min,\max)$-product of $A'$ and $B'$ equals to the minimum witness product of $A$ and $B$.
\end{proof}

\ThmApspToMinmax*
\begin{proof}
The reduction is similar to the reduction from \uapsp to \mono (Theorem~\ref{thm:ApspToMono}) in that we also have $\log n$ rounds, and in the $i$-th round we compute matrix $D^{\leq 2^i}$ of lengths of shortest paths of length up to $2^i$ (other entries equal to $\infty$). The key difference is that, in each round, instead of performing a binary search and issuing $\log n$ calls to \mono, we issue just two calls to \minmax.

As before, first note that $D^{\leq 2^0}$ is a $\{0,1,\infty\}$-matrix that can be easily obtained from the adjacency matrix of the input graph. Now, assume we already computed $D^{\leq 2^i}$ and let us proceed to compute $D^{\leq 2^{i+1}}$. Let $\ell = 2^i$. Naturally, $D^{\leq 2\ell}$ is the $(\min,+)$-product of $D^{\leq \ell}$ with itself, but this sole observation is not enough for our purposes. We will exploit the fact that $D^{\leq \ell}$ is not an arbitrary matrix -- but a (truncated) matrix of shortest paths in an unweighted graph -- in order to compute that specific $(\min,+)$-product using a \minmax algorithm. Let $A \owedge B$ denote the $(\min,\max)$-product of matrices $A$ and $B$.

First, we handle even-length paths. We compute $E = 2 \cdot (D^{\leq \ell} \owedge D^{\leq \ell})$. Note that $D^{\leq 2\ell}_{uv} \leq E_{uv}$ for all $u, v \in V$, because for any two integers $a, b$ we have $a + b \leq 2\cdot\max(a,b)$. Moreover, if $D^{\leq 2\ell}_{uv}=2k$, then there must exist $w \in V$ such that $D^{\leq\ell}_{uw} = D^{\leq\ell}_{wv} = k$, and thus $D^{\leq\ell}_{uw} + D^{\leq\ell}_{wv} = 2\cdot\max(D^{\leq\ell}_{uw}, D^{\leq\ell}_{wv})$ and $D^{\leq 2\ell}_{uv} = E_{uv}$.

For odd-length paths we proceed in a similar manner, just the formulas become slightly more obscure. We compute $O = 2 \cdot (D^{\leq \ell} \owedge (D^{\leq \ell} - 1)) + 1$. Note that $D^{\leq 2\ell}_{uv} \leq O_{uv}$ for all $u, v \in V$, because for any two integers $a, b$ we have $a + b \leq 2\cdot\max(a,b-1)+1$. Moreover, if $D^{\leq 2\ell}_{uv}=2k+1$, then there must exist $w \in V$ such that $D^{\leq\ell}_{uw} = k$ and $D^{\leq\ell}_{wv} = k + 1$, and thus $D^{\leq\ell}_{uw} + D^{\leq\ell}_{wv} = 2\cdot\max(D^{\leq\ell}_{uw}, D^{\leq\ell}_{wv}-1)+1$ and $D^{\leq 2\ell}_{uv} = O_{uv}$.

Consequently, we compute $D^{\leq 2\ell}_{uv} = \min(E_{uv}, O_{uv})$, for all $u,v \in V$.
\end{proof}

\section{Reductions for Convolution Problems}
 
In this section we provide two reductions which together show that \monoconv is fine-grained equivalent to $3$SUM. Recall that the best known algorithms for \monoconv require time $n^{3/2-o(1)}$, and the best algorithms for $3$SUM require time $n^{2-o(1)}$, so this is an equivalence between problems of different time complexity. At the end of the section we reduce \coin to \monoconv.

First, let us recall the All-Integers variant of $3$SUM, which parallels the All-Edges variants of our graph problems. That variant is easier to work with than the original $3$SUM problem for our purposes. Luckily if either variant has a subquadratic algorithm then they both do~\cite{Subcubic}.

\begin{definition}[All-Integers $3$SUM]
Given three lists $A, B, C$ of $n$ integers each, output the list of all integers $c \in C$ such that there exist $a \in A$ and $b \in B$ such that $a + b = c$.
\end{definition}

\begin{lemma}[Vassilevska Williams, Williams~\cite{Subcubic}]
\label{lem:allints3sum}
If $3$SUM is in $\Oh{n^{2-\eps}}$ time, then All-Integers $3$SUM is in $\Oh{n^{2-\eps/2}}$ time.
\end{lemma}

An important ingredient of our reduction from $3$SUM to \mono (Theorem~\ref{thm:3sumToConv}) is the following range reduction for $3$SUM.

\begin{lemma}[Baran, Demaine, P{\v{a}}tra{\c{s}}cu, rephrased, see Section~2.1 of~\cite{Baran08}]
For every positive integer output size $s$, there exists a family of hash functions $H$ such that:
\begin{enumerate}
 \item Every hash function $h\in H$ hashes to the range $\{0,1,\ldots,R-1\}$ for $R=2^s$.
 \item For all integers $a, b, c \in \mathbb{Z}$ and all hash functions $h \in H$, if $a+b=c$, then \[h(a)+h(b) \equiv h(c)+\{-1,0,1\} \mod R.\]
 \item Given an integer $c$ and two lists of $n$ integers $A$ and $B$ such that there are no $a \in A, b \in B$ with $a+b=c$, 
 the probability, over hash functions $h$ drawn uniformly at random from $H$, that there exist $a \in A, b \in B$ such that  
 $h(a)+h(b) \equiv h(c)+\{-1,0,1\} \mod R$ is at most 
 $\Oh{n^2/R}$. 
\end{enumerate}
\label{lem:BaranHashing}
\end{lemma}

We are now ready to show that $3$SUM can be solved efficiently with a \monoconv algorithm. Our reduction uses the fact that we can re-write a $3$SUM instance with $n$ integers in $\{-R,\ldots,R\}$ as a convolution of $\Oh{R}$-length $(0,1)$ vectors, where a one in the $i$-th position corresponds to the number $i$ in the original $3$SUM instance. We will combine several such instances into one \monoconv instance by giving each instance its own number. A one in position $i$ in a convolution instance labelled $j$ will result in the \monoconv instance having $j$ in position $i$.

\ThmTsumToConv*
\begin{proof}
Given an instance of $3$SUM of size $n$, we reduce it to $\Oh{n^{2/3}}$ instances of size $\Oh{n^{2/3}}$ each, using the self-reduction (Lemma~\ref{lem:3sumself}). Although for the self-reduction itself it would be sufficient just to solve $3$SUM on each of these instances  --  i.e.~decide if there exist $a,b,c$ with $a+b=c$  --  we are going to solve the All-Integers $3$SUM variant  --  i.e.~decide for each $c$ if there exist $a$ and $b$ with $a+b=c$.

To each created instance we apply a hashing scheme of Lemma~\ref{lem:BaranHashing} in order to reduce the universe size down to $R=n^{4/3}$. This introduces false positives for each element with probability $\Oh{(n^{2/3})^2/R}=\Oh{1}$. Note that the hashing has one-sided error, i.e.~if for some element $c$ there are no $a$ and $b$ such that $h(a)+h(b) \equiv h(c)+\{-1,0,1\} \mod R$, then with certainty there are no $a$ and $b$ such that $a+b=c$. To mitigate the effect of false positives we create $\Oh{\log n}$ copies of each instance, each copy using an independently drawn hash function. Note that for every fixed element $c$, if there are no $a$, $b$ with $a+b=c$, then the probability that in each of the independent $\Oh{\log n}$ copies we detect that $h(a)+h(b) \equiv h(c)+\{-1,0,1\} \mod R$ for some $h(a),h(b)$ is $1/\mathrm{poly}(n)$, and we can make the degree of the polynomial arbitrarily large by choosing an appropriate multiplicative constant for the number of copies. Therefore we can use the union bound to argue that with at least $2/3$ probability there are no false positives across all instances and all elements.

Suppose that for some (sub-)instance $A,B,C$ of size $n^{2/3}$ we learned, for every one of the $\Oh{\log n}$ hashed instance copies, for every $t\in [R]$ such that there is some $c$ with $h(c)=t$, whether there are some $h(a),h(b)$ with $h(a)+h(b) \equiv h(c)+\{-1,0,1\} \mod R$. Then, we can go through every $c\in C$ and if for every copy the answer for $h(c)$ was YES, we can conclude that (with high probability) there exists $a\in A, b\in B$ with $a+b=c$, and if the answer was NO at least once, then we can conclude that there is no pair that sums to $c$.

Here an important point is that we need to solve {\em all} of the $\Oh{n^{2/3}\log n}$ instances of All-Integers $3$SUM above, each on $n^{2/3}$ integers over a range $[\Oh{n^{4/3}}]$. We will embed solving all instances simultaneously into solving a small (polylogarithmic) number of \monoconv instances.

Each of the above $\Oh{n^{2/3}\log n}$ instances of All-Integers $3$SUM easily reduces to an $(\text{OR},\text{AND})$-convolution of $(0,1)$ vectors of length $\Oh{n^{4/3}}$, each with only $\Oh{n^{2/3}}$ nonzero entries, and with only $\Oh{n^{2/3}}$ relevant output coordinates one needs to compute. If only we had no collisions -- i.e.~two instances with the same nonzero input coordinate or the same relevant output coordinate -- we could easily combine all the convolution instances into a single instance of \monoconv, with $\Oh{n^{2/3}\log n}$ different colors/values. However, the collisions are unavoidable. In order to circumvent these collisions, we will add small random shifts, and use a similar analysis as in the $3$SUM-to-\mono reduction of Theorem~\ref{thm:3sumToMono}.

Specifically, for each $3$-SUM (sub-)instance we chose a shift $s$ uniformly at random from a range of size $\Oh{n^{4/3}}$, we add $s$ to all elements in $A$, add $s$ to all elements in $B$, and add $-2s$ to all elements in $C$. These shifts do not change whether for a given triplet $a, b, c$ the condition $a+b=c$ holds or not. Let the numbers after the shift lie in $\{-R',\ldots,R'\}$ where $R'=\Oh{n^{4/3}}$.
For a fixed value $v\in \{-R',\ldots,R'\}$ the expected number of instances containing $v$ is $\Oh{\log n}$.
Indeed, for each particular instance, the probability that one of its numbers lands at $v$ after the shift is $\Oh{n^{2/3}/R'} = \Oh{1/n^{2/3}}$; then summing over all the instances gives an expectation of $\Oh{\log n}$.

Since the shifts are independent, we can use a Chernoff bound to bound the probability that the number of instances containing $v$ exceeds $c \log n$ by $\leq (1/e)^{\Theta(c \log n)}$. We take $c$ large enough so that, by union bound, the probability that no value is contained in more than $c \log n$ instances is at least $2/3$. 

Then, once again following the example of Theorem~\ref{thm:3sumToMono}, we reduce the problem to $(c \log n)^3$ instances of \monoconv as follows.

For each value $r \in \{-R',\ldots,R'\}$, let the instances that contain $r$ in their $A$ sets be 
$in_A(r)[1],\ldots,$ $in_A(r)[c\log n]$. Define $in_B(r)[1],\ldots,in_B(r)[c\log n]$ and $in_C(r)[1],\ldots,in_C(r)[c\log n]$ analogously.

We now create an instance of $\monoconv$ for each choice of $(x,y,z)\in [c\log n]^3$. In instance $(x,y,z)$ we create vectors $a,b,c$, where for each $r\in \{-R',\ldots,R'\}$, we set $a_r=in_A(r)[x]$, $b_r=in_B(r)[y]$ and $c_r=in_C(r)[z]$. 
Then for any instance $i$ that contains $r$ in $A$, $s$ in $B$ and $t$ in $C$, we would have $in_A(r)[x]=in_B(s)[y]=in_C(t)[z]=i$ for some $x,y,z$ and so we will place $i$ in $a_r,b_s,$ and $c_t$ for that choice of $x,y,z$.
\end{proof}

This next reduction finishes the equivalence between \monoconv~and $3$SUM. It uses a high-frequency/low-frequency split. For elements that appear at a high frequency we use FFT. For elements of low frequency we make calls to All-Integers $3$SUM. Recall that a subquadratic algorithm for $3$SUM implies a subquadratic algorithm for All-Integers $3$SUM (Lemma~\ref{lem:allints3sum}).

\ThmConvToTsum*
\begin{proof}
For a parameter $t$ to be determined later, consider the $t$ most frequent values. For each of these values use FFT to calculate the standard $(+,\times)$-convolution of two $(0,1)$ vectors formed from vectors $a$ and $b$ by putting ones everywhere that value appears, and zeros everywhere else. Examine, where the outputs of these convolutions are nonzero, in order to determine the part of output to \monoconv corresponding to occurrences of the frequent values in vector $c$. This takes $\Ot{tn}$ time in total.

Let $n_i$ denote the number of occurrences of the $i$-th of the remaining values in all three sequences. Clearly, $\forall_i\ n_i \leq 3n/t$, and $\sum_i n_i \leq 3n$. For each value $v$ out of those remaining values construct sets of indices at which it appears in vectors $a$, $b$, $c$, i.e.~$A=\{j : a_j=v\}$, $B=\{j : b_j=v\}$, $C=\{j : c_j=v\}$, and solve All-Integers $3$SUM on these sets. For each element $j$ reported by the All-Integers $3$SUM algorithm assign the corresponding output of \monoconv $d_j=1$. By Lemma~\ref{lem:allints3sum}, solving these All-Integers $3$SUM instances takes an order of
\[\sum_i n_i^{2-\eps/2} = \sum_i n_i\cdot n_i^{1-\eps/2} \leq \sum_i n_i\cdot(3n/t)^{1-\eps/2} \leq 3n \cdot(3n/t)^{1-\eps/2}\]
time. The total time is thus $\Ot{tn + n \cdot(n/t)^{1-\eps/2}}$. Optimize by setting $t=n^{1/2-\eps/(8-2\eps)}$, and get the desired runtime.
\end{proof}

Our final theorem connects the \coin~problem to our network of reductions. The proof uses the same structure and techniques as the reduction from \uapsp to \mono in Theorem~\ref{thm:ApspToMono}.

\ThmCoinToConv*
\begin{proof}
Let $S$ denote the array of output values, i.e.~$S[v]$ equals to the minimum number of coins that sum to $v$. Parallel to the proof of Theorem~\ref{thm:ApspToMono}, let $S^{\leq 2^i}[v]$ be infinity if $S[v]>2^i$, and otherwise equal to $S[v]$. We solve \coin~in $\log n$ rounds, in the $i$-th round we compute $S^{\leq 2^i}$.

Note that $S^{\leq 2^0}[0] = 0$, and, for $v \geq 1$, $S^{\leq 2^0}[v] = 1$ if $v \in C$ and $S^{\leq 2^0}[v] = \infty$ otherwise. Further note that $S=S^{\leq 2^{\log n}}$. We will show how to compute $S^{\leq 2^{i+1}}$ given $S^{\leq 2^i}$. We will then iterate $i$ from $0$ up to $\log n$.

Following the style of Theorem~\ref{thm:ApspToMono}, we avoid overparameterizing by setting $A = S^{\leq 2^i}$ and $B = S^{\leq 2^{i+1}}$. Let $B^{(\ell)}$ be an array pointing to $2^\ell$-length intervals in which entries of $B$ lie, i.e.~$B^{(\ell)}[v]=j$ if $B[v] \in [2^{\ell}j,2^{\ell}(j+1)-1]$. We will iterate $\ell$ from $i+2$ down to $0$ to compute $B$ from $A$.

First we show how to compute $B^{(i+2)}$ from $A$. If there is a way to sum to $v$ with at most $2^{i+1}$ coins, then there must be a $u \in [0,n]$ such that both $A[u]$ and $A[v-u]$ are at most $2^i$. Conversely, if there is no way to sum to $v$ with at most $2^i$ coins, then there will be no $u$ that meets the above criteria. Therefore we create a $(0,1)$ vector $a$ with $a_v=1$ if and only if $A[v]$ is finite. Then, we compute the $(+,\times)$-convolution of $a$ with itself, in near-linear time using FFT. We set $B[v]=0$ where the convolution output is non-zero and $B[v]=\infty$ everywhere else.

Now we show how to compute $B^{(\ell)}$ from $A$ and $B^{(\ell+1)}$. Note that if $B^{(\ell+1)}[v] = j$ then $B^{(\ell)}[v]\in \{2j,2j+1\}$. Next, note that if $B^{(\ell)}[v] = 2j$, then there must exist an integer $u \in [0,n]$ such that 
\begin{equation}
\label{eqn:coin}
    A[u]\in \big[2^{\ell-1}\cdot (2j), 2^{\ell-1} \cdot (2j+1)\big), \quad\text{and}\quad A[v-u]\in \big[2^{\ell-1} \cdot (2j), 2^{\ell-1} \cdot (2j+1)\big].
\end{equation}

Furthermore, if $B^{(\ell)}[v]>2j$, then there is no $u$ such that the above condition holds. This will allow us to distinguish between the $2j$ and $2j+1$ cases. Note that the ranges in Condition~(\ref{eqn:coin}) do not overlap with corresponding ranges for different integer values $j' \ne j$. Thus, we will be able to use a single call to \monoconv to check in parallel for all values $v$ if $B^{(\ell)}[v]$ is the smaller even value $2B^{(\ell+1)}[v]$ or the larger odd value $2B^{(\ell+1)}[v]+1$.

We construct a \monoconv instance with three input vectors $a, b, c$.
The first input vector corresponds to the first part of Condition~(\ref{eqn:coin}), i.e.~if $A[v] \in \big[2^{\ell-1}\cdot (2j), 2^{\ell-1}\cdot(2j+1)\big)$, then $a_v=j$. Any entries $a_v$ unset by this condition are given the special value $a_v=-1$.
The second vector corresponds to the second part of Condition~(\ref{eqn:coin}), i.e.~if $A[v] \in \big[2^{\ell-1}\cdot (2j), 2^{\ell-1}\cdot(2j+1)\big]$, then $b_v=j$. Similarly, any entries $b_v$ unset by this condition are given the special value $b_v=-1$.
The last vector corresponds to our desired output, i.e.~$c_v = B^{(\ell+1)}[v]$.
Let $d$ denote the vector output by this \monoconv~call. 
Now, if $d_v=1$ then $B^{(\ell)}[v] = 2B^{(\ell+1)}[v]$, else  $B^{(\ell)}[v] = 2B^{(\ell+1)}[v]+1$. 

We iterate down until $B^{(0)}$, and observe that $B^{(0)}=B$. With $\Oh{\log n}$ calls to \monoconv we can thus compute $B=S^{\leq 2^{i+1}}$ from $A = S^{\leq 2^i}$. To solve \coin the total number of \monoconv calls is $\Oh{\log^2 n}$. Therefore if \monoconv can be solved in $T(n)$ time, then \coin can be solved in $\Oh{T(n)\log^2 n}$ time.
\end{proof}

\bibliography{main}

\end{document}